\documentclass[letterpaper, 10 pt, conference]{ieeeconf}  

\IEEEoverridecommandlockouts                              

\usepackage{graphicx} 
\usepackage{amsmath} 
\usepackage{amssymb}  

\usepackage{amsthm}

\usepackage{booktabs}
\usepackage{siunitx}
\usepackage{algorithm}
\usepackage{algpseudocode}

\newtheorem{assump}{Assumption}
\newtheorem{lem}{Lemma}

\newtheorem{thm}{Theorem}

\newtheorem*{remark}{Remark}

\title{\LARGE \bf
Stable Multi-Step Rollouts via Uncertainty-Guided Hybrid Dynamics*
}

\author{Andrei Maalberg$^{1}$, Axel Neumann$^{1}$ and Jens Knobloch$^{1,2}$
\thanks{*This work was supported by European Commission’s Horizon Europe Research and Innovation programme under Grant Agreement n°101131435}
\thanks{$^{1}$All authors are with Helmholtz-Zentrum Berlin, 14109 Berlin, Germany
        {\tt\small \{andrei.maalberg, axel.neumann, jens.knobloch\}@helmholtz-berlin.de}}%
\thanks{$^{2}$Jens Knobloch is also with the Department of Physics, Universität Siegen, 57068 Siegen, Germany {\tt\small jens.knobloch@uni-siegen.de}
        }%
}

\begin{document}

\maketitle
\thispagestyle{empty}
\pagestyle{empty}

\begin{abstract}

Multi-step rollouts are essential for model-based reinforcement learning (RL) and predictive control, yet learned dynamics models often become unstable when recursively applied, leading to divergence and unreliable policy updates. This paper proposes a model-agnostic hybrid dynamics framework that blends a provably contracting nominal model with a flexible excursion model through an uncertainty-guided switching law. The switching signal is derived from calibrated epistemic uncertainty and activates only when the system leaves the nominal region, ensuring that each model operates within its reliability regime. Under clearly stated smoothness and boundedness assumptions, we show that the resulting hybrid predictor yields globally bounded recursive multi-step rollouts: trajectories remain Lyapunov-stable in the nominal region and exhibit at most affine growth during excursions. To illustrate the theory in practice, we instantiate the hybrid dynamics framework within a model-based RL scheme that uses real one-step transitions for value learning and hybrid rollouts for policy improvement. Experiments on a nonlinear Duffing oscillator demonstrate stable long-horizon prediction and improved cost-effort trade-offs relative to a stabilizing baseline.

\end{abstract}

\section{INTRODUCTION}

Learning-based predictors are increasingly used in control and reinforcement learning (RL) for forecasting, planning, and multi-step decision making~\cite{Muthali_2023_uq, Zifan_2022_p2p, Balim_2025_imitation}. When recursively applied, however, small one-step errors can accumulate, or compound, leading to rollout drift and divergence from the true dynamics~\cite{Janner_2019_model_trust}. From a control-theoretic perspective, this reflects the absence of guarantees on contractivity, spectral properties, or behavior outside the training distribution~\cite{Kolter_2019_lyap}.

Recent work in imitation and RL mitigates compounding error by training direct multi-step predictors rather than recursively iterating a one-step model~\cite{Balim_2025_imitation, Li_2025_multi_step}. While this reduces statistical bias in long-horizon prediction, it sidesteps the question of stability under recursive application. In control settings---where simulation, filtering, and policy evaluation inherently rely on iterated dynamics---such recursive structure cannot be avoided and must be analyzed directly.

Within control theory, a classical remedy is to impose structural restrictions on the model class, for example through linear state-space approximations or Koopman-inspired embeddings of nonlinear dynamics~\cite{Lusch_2018_koopman}. These models enable analysis of long-horizon behavior but often lack the expressiveness needed to capture nonlinear transients or regime-dependent dynamics. In many practical systems---mechanical, biological, or radio frequency---dynamics alternate between nominal regions, where simple linear models suffice, and excursion regimes, where nonlinear effects dominate. A single monolithic learned model struggles to provide both stability in nominal operation and flexibility during excursions.

This paper addresses this tension by proposing a model-agnostic hybrid dynamics framework with uncertainty-triggered switching. The hybrid predictor consists of (i) a contractive nominal model governing stable operation and (ii) a bounded-growth excursion model capable of representing nonlinear transient behavior. A switching signal determines which model is trusted at each step, based on calibrated epistemic uncertainty. When uncertainty with respect to the nominal predictor is low, the dynamics remain in a contracting regime. Elevated uncertainty triggers a bounded excursion governed by a complementary nonlinear model. Once uncertainty decreases, the dynamics return to the nominal regime, where contraction is restored. This creates self-regulating learned dynamics, in which uncertainty acts as the mechanism enforcing stable and bounded recursive multi-step rollouts.

The main contribution of this paper is a theoretical characterization of stability and boundedness for this uncertainty-triggered hybrid predictor. Under clearly stated structural assumptions---contractivity of the nominal model, bounded growth of the excursion model, and calibration of uncertainty---we prove uniform global boundedness of recursive multi-step rollouts, without requiring multi-step prediction accuracy. Inside the nominal region, trajectories satisfy a Lyapunov stability guarantee; outside, they exhibit controlled affine growth before re-entering the nominal basin.

To illustrate the theory in a model-based RL setting, we integrate the hybrid predictor into a policy optimization framework, which we refer to as KIND-RL. In this instantiation, the predictor is implemented using the Kalman-Inspired Neural Decomposition (KIND) architecture introduced in our earlier work~\cite{Maalberg_2026_kind_kalman}. KIND combines a stable Koopman-like nominal predictor with an expressive sequence model for excursions augmented with epistemic uncertainty estimates. Although KIND serves as a concrete implementation within the proposed KIND-RL scheme, the theoretical results apply to any pair of models satisfying the stated assumptions.

Our contributions are:

\begin{enumerate}
  \item A hybrid-system formulation for analyzing uncertainty-triggered switching between contracting and bounded-growth learned dynamics.
  \item A Lyapunov-based stability lemma guaranteeing controlled behavior in the nominal region.
  \item A transient boundedness lemma ensuring finite-horizon safety during excursion phases.
  \item A global boundedness theorem for uncertainty-triggered hybrid rollouts.
  \item A model-based RL instantiation (KIND-RL) that integrates the proposed hybrid predictor into stable long-horizon policy improvement, illustrated on a nonlinear Duffing oscillator.
\end{enumerate}

This theoretical foundation supports applications in model-based RL, where stable recursive multi-step rollouts enable policy improvement beyond the one-step horizon traditionally imposed by the instability of learned models.

\section{RELATED WORK}

Recent work in learning-enabled control integrates data-driven models with classical stability analysis. The present work connects to three areas: (i) operator-theoretic and linear surrogate models, (ii) uncertainty-aware learning for control, and (iii) hybrid and switching dynamical systems.

\subsubsection{Operator-Theoretic and Linear Predictors}

Koopman-based and linear surrogate models have gained traction as reliable approximations of nonlinear dynamics, particularly in stationary regimes where linearization captures the dominant behavior. Approaches such as extended dynamic mode decomposition (DMD) and Koopman operators with neural observables~\cite{Lusch_2018_koopman} provide analyzable linear predictors whose propagation admits spectral or contraction-based guarantees. These models, however, struggle with transient or regime-changing behavior, where global linearity becomes inadequate. Our nominal model follows this operator-theoretic tradition, while our excursion model provides the expressive nonlinear correction needed to capture short-lived departures from stationary behavior. Related hybrid operator structures that combine global and local dynamics have also been explored in forecasting settings~\cite{Wang_2023}; in contrast, the present work analyzes an uncertainty-triggered blending mechanism and provides stability guarantees for the resulting hybrid predictor.

\subsubsection{Uncertainty in Learned Dynamics}

A significant amount of literature uses uncertainty estimates to regulate learning or control decisions. Disturbance learning based on a Gaussian process (GP) has been applied to safe control via reachability or robust model predictive control. For example, authors in~\cite{Akametalu_2014_reachability} and~\cite{Berkenkamp_2017_safe_mbrl} demonstrated that GP posterior variance enables selective conservativeness and safe exploration. Similarly, recent work~\cite{Das_2023_uq_decomposition} on uncertainty decomposition has used neural networks to separate aleatoric and epistemic uncertainty for adaptive data collection or feedback gain tuning. These methods assume known nominal dynamics and use uncertainty to filter control actions, tune feedback gains, or guide data acquisition, while the underlying model structure remains fixed. In contrast, our work employs uncertainty to regulate the model class itself, enabling bounded multi-step prediction in fully learned dynamics without assuming a known nominal system.

\subsubsection{Hybrid and Switching Models}

Hybrid systems traditionally assume a priori partitioned state spaces with sharp mode transitions. Analysis tools such as control Lyapunov functions, control barrier functions, or polynomial chaos-based uncertainty quantification provide guarantees under non-overlapping regime assumptions~\cite{Aali_2024_cbf, Sahai_2012_hybrid}. Recent extensions learn mode-dependent dynamics or switching surfaces using GPs~\cite{Beckers_2023_hamiltonian}. However, these frameworks require discrete mode boundaries and do not support soft blending. The hybrid dynamics studied in this paper differ fundamentally: switching is uncertainty-triggered, continuous, and designed specifically to stabilize long-horizon predictions rather than enforce safety or satisfy hard constraints. To the best of our knowledge, no prior work provides stability guarantees for learned multi-step rollouts in such "soft hybrid" systems.

\subsubsection{Summary}

Existing approaches provide either stable but inflexible linear predictors, uncertainty-aware control under known dynamics, or switching systems with discrete predefined modes. None provide stability guarantees for uncertainty-guided hybrid predictors in fully learned dynamics. The present work develops such guarantees and illustrates them using the KIND architecture, although the analysis is not specific to KIND.

\section{Model-Agnostic Stability Theory}\label{sec:theory}

This section develops a general stability framework for discrete-time hybrid systems whose switching is triggered by an epistemic-uncertainty signal. The results apply to any pair of (i) a contracting nominal model and (ii) a bounded-growth excursion model. They are later instantiated by the learned model used in Section~\ref{sec:app}, but the theory itself is entirely model-agnostic.

\subsection{Problem Setting and Notation}

We consider a discrete-time dynamical system\footnote{For the theoretical analysis, the control input may be generated by a fixed policy $u_t = \pi(x_t)$, but the notation $f(x_t, u_t)$ is kept for generality.}

\begin{equation*}
  x_{t+1} = f(x_t, u_t),
\end{equation*}

with state $x_t \in \mathbb{R}^n$ and control input $u_t \in \mathbb{R}^m$. The framework provides two predictors:

\begin{equation*}
  f^{\text{nom}},  \; f^{\text{exc}} \colon \mathbb{R}^{n+m} \to \mathbb{R}^n,
\end{equation*}

and a nonnegative epistemic-uncertainty estimate

\begin{equation*}
  \zeta(x_t, u_t) \in \mathbb{R}_{\geq 0}.
\end{equation*}

A switching law produces a blending coefficient

\begin{equation*}
  \alpha_t = \sigma(\zeta(x_t, u_t)) \in [0,1],
\end{equation*}

where $\sigma(\cdot)$ is any monotone function satisfying the regime-separation rule

\begin{equation*}
  \zeta(x_t, u_t) \leq \zeta^{\star} \Longrightarrow \alpha_t = 1, \quad \zeta(x_t, u_t) > \zeta^{\star} \Longrightarrow \alpha_t < 1,
\end{equation*}

for some threshold $\zeta^{\star} > 0$. The hybrid predictor evolves as

\begin{equation}
  x_{t+1} = \alpha_t f^{\text{nom}}(x_t, u_t) + (1-\alpha_t) f^{\text{exc}}(x_t, u_t).
  \label{eq:theory_not_predictor}
\end{equation}

Throughout the theoretical development, $x_t$ denotes the state generated by the hybrid predictor under recursive application. The true system dynamics are invoked only when defining the modeling error in Assumption~\ref{thm:theory_a3}; all subsequent theoretical results concern the recursively applied hybrid predictor. In Section~\ref{sec:app}, the same predictor is embedded within the RL algorithm, where real system transitions are used for value learning and predictor states are used for model rollouts.

\subsection{Assumptions}

\begin{assump} [nominal-model contraction]
  \label{thm:theory_a1}

  There exists a continuously differentiable Lyapunov function

  \begin{equation*}
    V \colon \mathbb{R}^n \to \mathbb{R}_{\geq 0},
  \end{equation*}

  and constants $c_1, c_2, c_3 > 0$ such that

  \begin{equation}
    c_1 \lVert x \rVert ^2 \leq V(x) \leq c_2 \lVert x \rVert ^2,
    \label{eq:theory_a_lyap_quad}
  \end{equation}

  and

  \begin{equation}
    V(f^{\text{nom}}(x_t, u_t)) - V(x_t) \leq -c_3 \lVert x_t \rVert ^2,
    \label{eq:theory_a_lyap_dec}
  \end{equation}

  whenever $\zeta(x_t, u_t) \leq \zeta^{\star}$. Since $V$ is continuously differentiable, it is locally Lipschitz. Accordingly, on every compact Lyapunov sublevel set there exists a corresponding Lipschitz constant $L_V$. Thus the nominal subsystem admits a strict quadratic Lyapunov decrease in its confidence region.

\end{assump}

\begin{assump} [excursion-model bounded growth]
  \label{thm:theory_a2}

  There exists $M \geq 1$ such that for all $x_t$,

  \begin{equation*}
    \lVert f^{\text{exc}}(x_t, u_t) \rVert \leq M(1 + \lVert x_t \rVert).
  \end{equation*}

  Thus excursions may grow, but only at most linearly.
\end{assump}

\begin{assump} [uncertainty-calibrated modeling error]
  \label{thm:theory_a3}

  Let the modeling error of the nominal predictor be

  \begin{equation*}
    e(x_t, u_t) := f(x_t, u_t) - f^{\text{nom}}(x_t, u_t).
  \end{equation*}

  There exists a class-$\mathcal{K}$ function $\delta(\cdot)$ such that

  \begin{equation*}
    \lVert e(x_t, u_t) \rVert \leq \delta(\zeta(x_t, u_t)).
  \end{equation*}

  Thus uncertainty upper-bounds the model error.
\end{assump}

\begin{remark}
  In practice, uncertainty estimators may be trained by supervised regression and therefore provide learned proxies rather than certified upper bounds on the modeling error. Consequently, Assumption~\ref{thm:theory_a3} should be interpreted as an idealized calibration condition. If the calibration error is uniformly bounded, i.e.,

  \begin{equation*}
     \lVert e(x_t, u_t) \rVert \leq \delta(\zeta(x_t, u_t)) + \eta,
  \end{equation*}

  for some $\eta \ge 0$, the subsequent boundedness analysis remains qualitatively unchanged, with only the size of the resulting practical neighborhood increasing.
\end{remark}

\begin{assump} [finite excursion recovery]
  \label{thm:theory_a4}

  Under the closed-loop control policy, every interval during which $\zeta(x_t,u_t) > \zeta^\star$
has uniformly bounded duration. That is, there exists $H<\infty$ such that

  \begin{equation*}
    \zeta(x_t,u_t) > \zeta^\star \;\Rightarrow\; \exists \, k \le H: \zeta(x_{t+k},u_{t+k}) \le \zeta^\star.
  \end{equation*}

  Moreover, the closed-loop dynamics recover to the nominal-confidence region after each excursion, so that excursions remain transient and do not accumulate into persistent growth.
\end{assump}

\begin{remark}
  Assumption~\ref{thm:theory_a4} permits infinitely many switching events; it only requires that each excursion is transient and followed by recovery to the nominal-confidence region. This is analogous to bounded dwell-time assumptions in switched systems and models disturbances or regime transitions that eventually decay or are rejected by the closed-loop controller.
\end{remark}

\begin{assump} [bounded-gain control policy]
  \label{thm:theory_a5}
  
  Let the control input applied to the system be
  \begin{equation*}
    u_t = \pi(x_t).
  \end{equation*}
 
 There exist constants $c_u > 0$ and $c_0 \geq 0$ satisfying
  
  \begin{equation*}
    \lVert u_t \rVert \leq c_u \lVert x_t \rVert + c_0, \quad \forall t.
  \end{equation*}

\end{assump}

\subsection{Lemmas and Theorem}

We present the main lemmas and the global boundedness theorem below. All proofs are deferred to Appendix.

\begin{lem} [practical nominal contraction]
  \label{thm:theory_lemma1}

  Under Assumptions~\ref{thm:theory_a1} and~\ref{thm:theory_a3}, if $\zeta(x_t,u_t) \le \zeta^\star$, then

  \begin{equation*}
    V(x_{t+1}) - V(x_t) \leq - c_3 \lVert x_t \rVert ^2 + \varepsilon,
  \end{equation*}

  where $\varepsilon := L_V \delta(\zeta^\star)$. In particular, there exists a radius

  \begin{equation*}
    r := \sqrt{\frac{2\varepsilon}{c_3}}
  \end{equation*}

  such that whenever $ \lVert x_t \rVert \geq r$,

  \begin{equation*}
    V(x_{t+1}) - V(x_t) \leq -\frac{1}{2} \, c_3 \, \lVert x_t \rVert^2.
  \end{equation*}
  
  Thus, within the nominal-confidence region, the true dynamics exhibit strict Lyapunov decrease outside a compact neighborhood of the origin.
  \end{lem}

\begin{lem} [finite-horizon boundedness]
  \label{thm:theory_lemma2}
  Under Assumption~\ref{thm:theory_a2}, for any sequence of $k \leq H$ consecutive excursion steps,

  \begin{equation*}
    \lVert x_{t+k} \rVert \, \leq \, M^k \lVert x_t \rVert \,+\, M \frac{M^k - 1}{M - 1}.
  \end{equation*}

  Thus even if the excursion dynamics do not satisfy contraction, their growth is uniformly bounded over any finite horizon determined by $H$.
\end{lem}

\begin{thm} [global boundedness under uncertainty-triggered switching]
  \label{thm:theory_theorem1}
  Under Assumptions~\ref{thm:theory_a1}--\ref{thm:theory_a5}, all trajectories of the hybrid system defined by (\ref{eq:theory_not_predictor}) satisfy

  \begin{equation*}
    \sup_t \, \lVert x_t \rVert < \infty.
  \end{equation*}

  Thus the uncertainty-triggered hybrid system is globally bounded, i.e., there exists a finite constant $C > 0$ such that $\|x_t\| \leq C$ for all $t$.
\end{thm}

\subsection{Interpretation}

Uncertainty acts as a state-dependent switching signal.

\begin{itemize}
  \item When the nominal model is confident: $V$ decreases $\rightarrow$ contraction dominates.
  \item When uncertainty rises: switching occurs $\rightarrow$ excursions may expand but only for a uniformly bounded duration.
\end{itemize}

Consequently, the hybrid trajectory remains uniformly bounded.

\section{Application to Model-Based RL Using KIND}\label{sec:app}

\subsection{Instantiating the Hybrid Predictor with KIND}

KIND maintains two learned predictors~\cite{Maalberg_2026_kind_kalman} corresponding to the two behavioral regimes assumed in Section~\ref{sec:theory}:

\begin{enumerate}
  \item Nominal model $f^{\text{nom}}(x_t,u_t)$: a locally-linear Koopman/DMD operator trained exclusively on in-distribution (ID) data, i.e., trajectories around the equilibrium and normal operating conditions.

  \item Excursion model $f^{\text{exc}}(x_t,u_t)$: a nonlinear Transformer trained only on out-of-distribution (OOD) segments such as impacts, kicks, limit-cycle oscillations, or abrupt behavioral changes.
\end{enumerate}

The nominal model provides strong local accuracy and inherent stability, while the excursion model captures nonlinear transients and regime-dependent effects beyond the representational capacity of a linear operator. Both nominal and excursion predictors operate on a finite lookback window of past state-control pairs. During recursive rollouts, this window is updated in a sliding fashion and periodically re-anchored to real replay states to prevent drift accumulation.

To drive the switching logic, KIND associates each predictor $f^b \mid b \in \{ \text{nom}, \text{exc} \}$ with a learned epistemic uncertainty module

\begin{equation*}
  \zeta^b(x_t, u_t) \in \mathbb{R}_{\geq 0}.
\end{equation*}

Each module is trained to predict the magnitude of the one-step modeling error of its corresponding predictor,

\begin{equation*}
  \zeta^b(x_t, u_t) \approx \lVert x_{t+1} - f^b(x_t, u_t) \rVert,
\end{equation*}

using mixed ID and OOD data during supervised training. At runtime, the uncertainty module provides a scalar confidence signal without access to the true next state. Thus, $\zeta^b$ acts as a learned proxy for epistemic modeling error, enabling regime selection during recursive rollouts.

Section~\ref{sec:theory} assumed a general switching map $\alpha_t = \sigma(\zeta(x_t, u_t))$ that selects the nominal model when uncertainty is low. KIND instantiates this switching using uncertainty-normalized blending:

\begin{equation*}
  \alpha_t = \frac{\zeta^{\text{exc}}_t}{\zeta^{\text{exc}}_t + \zeta^{\text{nom}}_t}\; \in [0,1].
\end{equation*}

This normalized uncertainty weighting assigns higher authority to the model whose associated uncertainty is lower, without assuming probabilistic noise models or Gaussian structure. This choice has three desirable properties:

\begin{enumerate}
  \item Nominal dominance in ID regions. If $\zeta_t^{\text{nom}} \ll \zeta_t^{\text{exc}}$, then $\alpha_t \approx 1$, recovering the contraction conditions required by Assumption~\ref{thm:theory_a1}.

  \item Excursion dominance in OOD regions. If the nominal uncertainty spikes, $\alpha_t$ decreases and $f^{\text{exc}}$ takes over.

  \item Suppression of spurious regime switching. If the excursion model is uncertain in an ID region, its large $\zeta^{\text{exc}}$ drives $\alpha_t\rightarrow 1$, keeping the system within the stable nominal dynamics.
\end{enumerate}

Thus the uncertainty-normalized blending implements the theoretical regime-separation logic automatically, without requiring manually tuned switching thresholds.

Control is implemented as a bounded residual policy on top of a fixed stabilizing baseline:

\begin{equation*}
  u_t = u_t^{\text{LQR}} + \Delta u_t^{\text{RL}},
\end{equation*}

where $u_t^{\text{LQR}} = -K x_t$ is a stabilizing LQR controller and $\Delta u_t^{\text{RL}} = \Delta \pi_{\phi}(x_t)$ is a learned residual term parametrized by $\phi$. The residual is explicitly bounded by design:

\begin{equation*}
  \lVert \Delta u_t^{\text{RL}} \rVert \leq u_{\text{max}}.
\end{equation*}

As a result, the applied control satisfies

\begin{equation*}
  \lVert u_t \rVert \leq \lVert K \rVert \lVert x_t \rVert + u_{\text{max}},
\end{equation*}

which corresponds to Assumption~\ref{thm:theory_a5} with $c_u = \lVert K \rVert$ and $c_0 = u_{\text{max}}$.

The resulting hybrid predictor

\begin{equation*}
  x_{t+1} = \Gamma(x_t, u_t) = \alpha_t f^{\text{nom}}(x_t, u_t) + (1 - \alpha_t) f^{\text{exc}}(x_t, u_t),
\end{equation*}

obtained by instantiating the model-agnostic framework (\ref{eq:theory_not_predictor}) of Section~\ref{sec:theory}, serves as the learned dynamics model for model-based planning in Section~\ref{sec:app_rl}. Table~\ref{tab:app_kind_theory_map} summarizes how the abstract assumptions introduced in Section~\ref{sec:theory} are realized by the KIND architecture. While several assumptions are enforced by construction, others represent modeling assumptions or approximate properties of the learned components.

\begin{table}[t]
  \caption{Correspondence between the theoretical assumptions of Section~\ref{sec:theory} and their realization in KIND.}
  \label{tab:app_kind_theory_map}
  \centering
  \begin{tabular}{p{0.18\linewidth} p{0.52\linewidth} p{0.15\linewidth}}
  \toprule
Theory & KIND realization & Nature \\
  \midrule
Assumption~\ref{thm:theory_a1}
& Koopman-like nominal predictor governing normal operation
& By design \\

Assumption~\ref{thm:theory_a2}
& Transformer excursion model with bounded-growth architecture
& By design \\

Assumption~\ref{thm:theory_a3}
& Learned uncertainty estimator providing a proxy for modeling error
& Approximate \\

Assumption~\ref{thm:theory_a4}
& Transient switching between nominal and excursion regimes
& Operating assumption \\

Assumption~\ref{thm:theory_a5}
& Residual RL on top of stabilizing LQR baseline
& By design \\
  \bottomrule
  \end{tabular}
\end{table}

\subsection{Stable Policy Improvement with Hybrid Rollouts}\label{sec:app_rl}

The hybrid predictor $\Gamma(x, u)$ defined above enables multi-step model-based policy improvement while preserving the stability properties established in Section~\ref{sec:theory}. We adopt a two-phase learning procedure which separates value estimation from policy improvement. Importantly, value learning uses only real transitions, whereas model rollouts are used exclusively for policy improvement. This separation prevents model bias from contaminating value estimation, while still allowing long-horizon credit assignment during policy improvement.

The value function $V_\theta(x)$ is trained on real one-step transitions collected from the true system. For a replay sample $(x_t, u_t, r_t, x_{t+1})$, a temporal-difference (TD) target is

\begin{equation*}
  y_t = r_t + \gamma V_\theta(x_{t+1}),
\end{equation*}

where $r_t = \mathcal{R}(x_t, u_t)$ denotes an instantaneous reward and $\gamma \in (0,1)$ is a discount factor. Parameters $\theta$ are updated by minimizing the squared Bellman residual

\begin{equation*}
  \mathcal{L}_V = \| V_\theta(x_t) - y_t \|^2.
\end{equation*}

The TD update corresponds to regression toward the fixed point of the Bellman operator associated with the current fixed policy.

In contrast, the $\Delta \pi_{\phi}(x)$ policy updates are computed using multi-step model rollouts. Starting from a replay state $x_t$, we generate a closed-loop $H$-step rollout using $\Gamma(x, u)$ and the current policy

\begin{equation*}
  \pi(x) = -Kx + \Delta \pi_\phi(x),
\end{equation*}

such that the predicted states satisfy

\begin{equation*}
  x_{t+k+1} = \Gamma(x_{t+k}, \pi(x_{t+k})).
\end{equation*}

The $H$-step model rollouts are used to estimate advantage

\begin{equation*}
  A_t^{(H)} = \sum_{k=0}^{H-1} \gamma^k \mathcal{R} (x_{t+k}, \pi(x_{t+k}))
  + \gamma^H V_\theta(x_{t+H})
  - V_\theta(x_t).
\end{equation*}

Policy parameters $\phi$ are updated by ascending the empirical advantage, implemented as minimizing a surrogate loss

\begin{equation*}
  \mathcal{L}_{\Delta\pi} = - \mathbb{E}\big[ A_t^{(H)} \big].
\end{equation*}

This corresponds to a deterministic policy-gradient step using hybrid model rollouts for multi-step credit assignment. Since the hybrid predictor satisfies global boundedness, the $H$-step rollout remains uniformly bounded for any finite horizon, ensuring well-defined policy gradients. In practice though, rollouts are periodically re-initialized from replay states to limit the accumulation of prediction bias during policy optimization, following common practice in model-based RL.

\section{EXPERIMENTS}\label{sec:exp}

\subsection{Duffing Oscillator Setup}

We consider the controlled Duffing oscillator

\begin{align*}
  \dot{x}_1 &= x_2, \\
  \dot{x}_2 &= -\kappa x_2 + \mu x_1 - \beta x_1^3 + \rho \cos(\omega t) + u,
\end{align*}

with parameters specified in Table~\ref{tab:exp_duffing_param}.

\begin{table}[!htb]
  \caption{Duffing parameters.}
  \label{tab:exp_duffing_param}
    \centering
    \begin{tabular}{crl}
        \toprule
        Parameter        & Value & Description  \\
        \midrule
        $\mu$           & 110.0 & linear stiffness \\
        $\beta$             & 140.0 & nonlinear stiffness \\
        $\rho$        & 70.0   & amplitude of driving force \\
        $\kappa$            & 0.5 & damping \\
        $\omega$            & 1.2 & angular frequency of driving force \\
        $\Delta t$            & 0.01 & discretization step \\
        \bottomrule
    \end{tabular}
\end{table}

The control objective is regulation toward the reference point $x^\star = [1, 0]^\top$. The reward is defined as
\begin{equation*}
  \mathcal{R}(x,u) = - (x-x^\star)^\top Q \, (x-x^\star) - u^\top R \,u,
\end{equation*}

where $Q=I$ and $R=\lambda I$. A stabilizing linear-quadratic regulator (LQR) is designed around the linearization at $x^\star$ using perturbed stiffness parameters to introduce controlled model mismatch.
The LQR is not intended as a performance benchmark, but rather serves two structural purposes:

\begin{enumerate}
  \item It guarantees a stabilizing baseline satisfying the bounded-gain condition of Assumption~\ref{thm:theory_a5}.

  \item It provides a nominal control signal for training the hybrid predictor in both nominal and excursion regimes.
\end{enumerate}

The LQR gains are selected to avoid fully suppressing inter-well transitions, thereby preserving the nonlinear switching behavior required to evaluate the uncertainty-driven hybrid architecture.

\subsection{Hybrid Rollout Stability}

We evaluate the long-horizon rollout behavior of the KIND hybrid predictor. The result is shown in Figure~\ref{fig:exp_rollout_kind}. The rollout spans 600 time steps. Due to the finite lookback structure of the predictor, the internal context window is re-anchored to real states every 20 steps.

Figure~\ref{fig:exp_rollout_kind}(a) shows the true Duffing trajectory and the recursive hybrid prediction. The predicted trajectory remains globally bounded and tracks the qualitative regime transitions without divergence, despite inter-well excursions.

Figure~\ref{fig:exp_rollout_kind}(b) depicts the nominal uncertainty $\zeta^{\text{nom}}$. Uncertainty remains low inside the nominal basin around $x^\star$, increases sharply during inter-well transitions, and decreases again upon re-entry into the nominal region.

Figure~\ref{fig:exp_rollout_kind}(c) shows the blending coefficient $\alpha$. As nominal uncertainty rises, authority shifts from the nominal predictor toward the excursion model. When the system returns to the nominal basin, $\alpha$ recovers toward unity, restoring contraction.

These observations empirically support the regime-separation mechanism of Section~\ref{sec:theory} and demonstrate bounded multi-step prediction under uncertainty-guided switching.
 
\begin{figure}[!htb]
    \centering
    \includegraphics[width=0.84\linewidth]{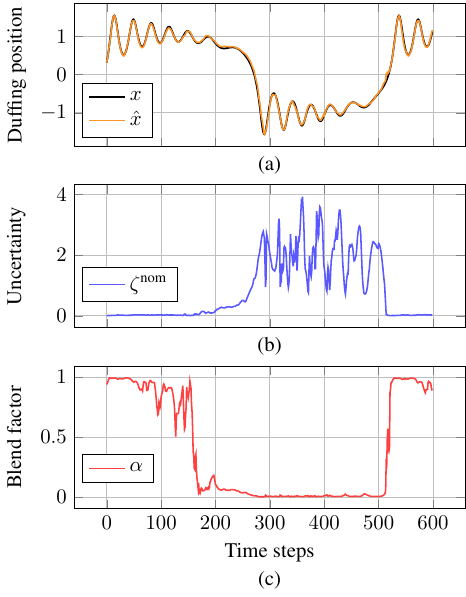}
    \caption{Long-horizon hybrid rollout with re-anchoring every 20 steps. (a) True Duffing position $x$ and recursive KIND prediction $\hat{x}$ over 600 steps. (b) Nominal uncertainty $\zeta^{\text{nom}}$. (c) Blending coefficient $\alpha$.
Uncertainty spikes during inter-well excursions, triggering a temporary shift toward the excursion model,
while predictions remain globally bounded.}
    \label{fig:exp_rollout_kind}
\end{figure}

\subsection{RL with Hybrid Rollouts}

We next demonstrate that the proposed predictor can be used for multi-step policy improvement within a model-based RL framework. Figure~\ref{fig:exp_replay_kindrl} compares the underlying LQR baseline with the residual policy obtained using hybrid rollouts with $H=200$. Relative to the baseline, the learned residual modifies the intra-well oscillation patterns. In particular, oscillations around the reference well $x^\star$ become shorter in duration, while oscillations in the opposite potential well are slightly prolonged. At the same time, control oscillations are visibly attenuated.

\begin{figure}[!htb]
    \centering
    \includegraphics[width=0.84\linewidth]{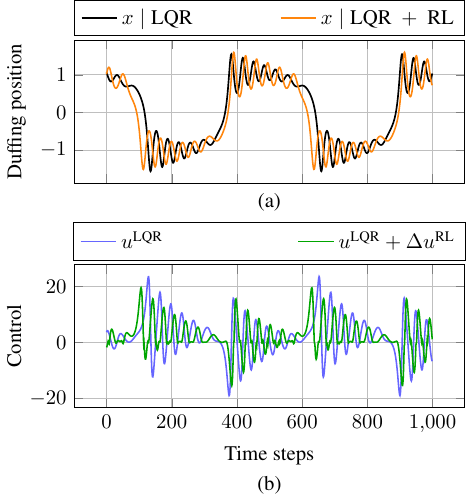}
    \caption{Behavior of residual RL policy during a replay compared to baseline LQR. (a) Residual policy makes dwell time in reference well $x^\star$ shorter compared to the opposite well. (b) Control effort is attenuated when residual policy augments LQR.}
    \label{fig:exp_replay_kindrl}
\end{figure}

The influence of rollout horizon is illustrated in Figure~\ref{fig:exp_replay_advantage} for $H \in \{2, 200\}$. Increasing $H$ changes the scale of the estimated advantage while preserving stability. Longer horizons produce smoother policy updates and reduced control oscillations, consistent with longer-horizon credit assignment.

\begin{figure}[!htb]
    \centering
    \includegraphics[width=0.84\linewidth]{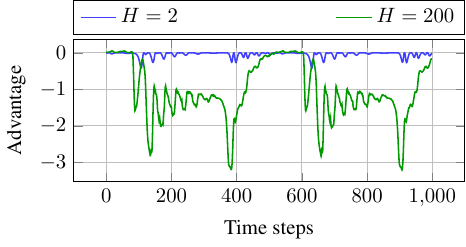}
    \caption{Estimated advantage along a replay trajectory for short ($H=2$) and long ($H=200$) hybrid rollouts. Increasing the rollout horizon substantially increases the magnitude and smoothness of the advantage signal, while preserving stability of the underlying trajectory.}
    \label{fig:exp_replay_advantage}
\end{figure}

Table~\ref{tab:exp_lqr_vs_kindrl} illustrates that the proposed hybrid predictor supports effective policy improvement, yielding lower cumulative cost and control effort than the stabilizing baseline. The cumulative cost refers to the finite-horizon sum of the quadratic stage cost used for policy evaluation. Meanwhile, dwell denotes average residence duration (in seconds) within the reference well. Hybrid policy results are reported as mean $\pm$ one standard deviation over five random seeds. During these experiments, the KIND predictor was fixed; thus randomness arose only from the RL training process.

\begin{table}[!htb]
  \caption{Reduction of cumulative cost, effort and dwell by KIND-RL.}
  \label{tab:exp_lqr_vs_kindrl}
    \centering
    \begin{tabular}{llll}
        \toprule
        Policy         & Cost & Effort & Dwell  \\
        \midrule
        LQR                               & $5.01$ & $7.22$ & $34.8$ \\
        LQR + RL : H=2            & $3.52 \pm 0.02$ & $5.97 \pm 0.02$ & $28.6 \pm 1.6$ \\
        LQR + RL : H=200        & $3.50 \pm 0.01$ & $5.95 \pm 0.01$ & $26.9 \pm 0.5$ \\
        \bottomrule
    \end{tabular}
\end{table}

\section{CONCLUSION}

We proposed a model-agnostic hybrid dynamics framework that blends a contracting nominal model with a flexible excursion model through an uncertainty-guided switching law. Under clearly stated structural assumptions, we established global boundedness of the resulting predictor and characterized conditions under which recursive multi-step rollouts remain well-defined.

The framework was instantiated within a model-based RL scheme, where the hybrid predictor serves as a learned dynamics model for stable long-horizon policy improvement. Experiments on a nonlinear Duffing oscillator demonstrated bounded regime transitions, adaptive uncertainty-guided blending, and improved cost-effort trade-offs relative to a stabilizing baseline.

The proposed formulation provides a principled mechanism for combining local stability guarantees with the modeling flexibility required to capture nonlinear transient dynamics. Future work will investigate adaptive uncertainty calibration, stability and performance guarantees under online policy adaptation, and extensions to higher-dimensional and partially observed systems.

\section*{Declaration of AI assistance}

The OpenAI's ChatGPT~\cite{OpenAI_ChatGPT} was used to improve the syntax and grammar of several paragraphs in the manuscript.


\bibliographystyle{IEEEtran}  
\bibliography{root}

\section*{APPENDIX}\label{sec:appendix}

\begin{proof}[Proof of Lemma~\ref{thm:theory_lemma1}]

  Let the true closed-loop dynamics be
  
  \begin{equation*}
    x_{t+1} = f(x_t, u_t) = f^{\text{nom}}(x_t, u_t) + e(x_t, u_t),
  \end{equation*}
  
  where $e(x_t, u_t)$ is the modeling error of the nominal predictor. Assume $\zeta(x_t, u_t) \leq \zeta^{\star}$, i.e., the system is in the nominal-confidence region. We evaluate the Lyapunov difference along the true dynamics:

  \begin{equation}
    \begin{split}
      V(x_{t+1}) - V(x_t) = V(f^{\text{nom}}(x_t, u_t) + e(x_t, u_t)) - V(x_t) \\
      = \underbrace{V(f^{\text{nom}}(x_t, u_t)) - V(x_t)}_{\text{nominal contraction}} \\
      +
      \underbrace{V(f^{\text{nom}}(x_t, u_t) + e(x_t, u_t)) - V(f^{\text{nom}}(x_t, u_t))}_{\text{error-induced deviation}}.
    \end{split}
    \label{eq:app_lemma1_split}
  \end{equation}
  
  This is simply an 'add and subtract $f^{\text{nom}}(x_t, u_t)$' trick that isolates nominal contraction from modeling error.

  Assumption~\ref{thm:theory_a1} gives, whenever $\zeta(x_t, u_t) \leq \zeta^{\star}$:

  \begin{equation}
    V(f^{\text{nom}}(x_t, u_t)) - V(x_t) \leq -c_3 \lVert x_t \rVert^2.
    \label{eq:app_lemma1_contract}
  \end{equation}

  Moreover, since $V \in C^1$, it is locally Lipschitz. On compact Lyapunov sublevel sets, a uniform Lipschitz constant $L_V$ exists:
  
  \begin{equation*}
    |V(a) - V(b)| \leq L_V \lVert a-b  \rVert.
  \end{equation*}

  Applying this to the error term in (\ref{eq:app_lemma1_split}):

  \begin{equation}
    \begin{split}
      V(f^{\text{nom}}(x_t, u_t) + e(x_t, u_t)) - V(f^{\text{nom}}(x_t, u_t)) \\
      \leq L_V \lVert e(x_t, u_t) \rVert.
    \end{split}
    \label{eq:app_lemma1_lv_bound}
  \end{equation}

  Assumption~\ref{thm:theory_a3} gives:

  \begin{equation*}
    \lVert e(x_t, u_t) \rVert \leq \delta(\zeta(x_t, u_t)),
  \end{equation*}

  but since we are in the nominal regime:

  \begin{equation}
    \zeta(x_t, u_t) \leq \zeta^{\star} \quad \Rightarrow \quad \lVert e(x_t, u_t) \rVert \leq \delta(\zeta^{\star}).
    \label{eq:app_lemma1_error_bound}
  \end{equation}

  Substitute (\ref{eq:app_lemma1_contract}), (\ref{eq:app_lemma1_lv_bound}) and (\ref{eq:app_lemma1_error_bound}) into (\ref{eq:app_lemma1_split}):

  \begin{equation}
    V(x_{t+1}) - V(x_t) \leq -c_3 \lVert x_t \rVert^2 + \varepsilon, \quad \varepsilon := L_V \delta(\zeta^{\star}).
    \label{eq:app_lemma1_combined}
  \end{equation}
  
  Therefore, whenever

  \begin{equation*}
     \lVert x_t \rVert ^2 \geq \frac{2 \varepsilon}{c_3},
  \end{equation*}

  we obtain

  \begin{equation*}
    V(x_{t+1}) - V(x_t) \leq -\frac{1}{2} \, c_3 \, \lVert x_t \rVert^2.
  \end{equation*}

  Thus the true dynamics are strictly contracting outside a ball of radius
  
  \begin{equation*}
    r := \sqrt{\frac{2 \varepsilon}{c_3}}.
  \end{equation*}
  
  Inside this ball, boundedness follows from continuity.

\end{proof}

\begin{proof}[Proof of Lemma~\ref{thm:theory_lemma2}]
  Assumption~\ref{thm:theory_a2} states that the excursion model satisfies
  
  \begin{equation}
    \lVert f^{\text{exc}}(x_t, u_t) \rVert \leq M (1 + \lVert x_t \rVert), \quad M \geq 1.
    \label{eq:app_lemma2_bound}
  \end{equation}

  Consider an uninterrupted sequence of $k \le H$ excursion steps and let

  \begin{equation*}
    z_k := \lVert x_{t+k} \rVert.
  \end{equation*}

  Then (\ref{eq:app_lemma2_bound}) implies the affine recursion

  \begin{equation*}
    z_{k+1} \leq M + M z_k.
  \end{equation*}

  Unrolling:

  \begin{flalign*}
    \quad z_1 \leq& \;M + M z_0,&&\\
    \quad z_2 \leq& \;M + M z_1 \leq M + M (M + M z_0),&&\\
    \quad z_3 \leq& \;M + M z_2 \leq M + M (M + M (M + M z_0)),&&\\
    \vdots& &&
  \end{flalign*}

  Collecting terms yields

  \begin{equation*}
    z_k \le M^k z_0 + M \sum_{i=0}^{k-1} M^i. 
  \end{equation*}

  Using the formula for the finite geometric series,
  \begin{equation*}
    \sum_{i=0}^{k-1} M^i = \frac{M^k - 1}{M - 1},
  \end{equation*}

  we obtain

  \begin{equation*}
    z_{k} \, \leq \, M^k z_0 \,+\, M \frac{M^k - 1}{M - 1}.
  \end{equation*}

  Thus for any excursion of duration $k \leq H$,

  \begin{equation*}
    \lVert x_{t+k} \rVert \, \leq \, M^k \lVert x_t \rVert \,+\, M \frac{M^k - 1}{M - 1}.
  \end{equation*}
  
  Since $H < \infty$, the right-hand side is finite, establishing finite-time boundedness of excursions.
  
  For the special case $M=1$, the recursion reduces to $z_{k+1} \le 1 + z_k$, yielding $z_k \le z_0 + k$, which is also finite for $k \le H$.

\end{proof}

\begin{proof}[Proof of Theorem~\ref{thm:theory_theorem1}]
  We combine the results of Lemmas~\ref{thm:theory_lemma1} and~\ref{thm:theory_lemma2}.

  \subsection*{1. Dynamics inside the nominal region}

  By Lemma~\ref{thm:theory_lemma1}, there exists a radius

  \begin{equation*}
    r := \sqrt{\frac{2 \varepsilon}{c_3}}, \quad \varepsilon = L_V \delta(\zeta^{\star}),
  \end{equation*}
  
  such that whenever $\lVert x_t \rVert \geq r$,

  \begin{equation*}
    V(x_{t+1}) - V(x_t) \leq -\frac{1}{2} \, c_3 \, \lVert x_t \rVert^2 < 0.
  \end{equation*}

  Hence, along every nominal step outside the compact ball $\lVert x_t \rVert < r$, the Lyapunov function decreases strictly. Consequently, during every nominal phase the Lyapunov function decreases until the trajectory reaches a compact neighborhood of the origin. Inside the ball $\lVert x_t \rVert < r$, boundedness is immediate.

  \subsection*{2. Dynamics during excursion phases}

  When $\zeta(x_t, u_t) > \zeta^{\star}$, the system follows the excursion model $f^{\text{exc}}$. Lemma~\ref{thm:theory_lemma2} implies that for any excursion of length $k \leq H$,

  \begin{equation*}
    \lVert x_{t+k} \rVert \, \leq \, M^k \lVert x_t \rVert \,+\, M \frac{M^k - 1}{M - 1}.
  \end{equation*}

  Thus excursions may increase the state norm but only by a uniformly bounded finite amount.

  \subsection*{3. Concatenation of contracting and expansive intervals}
  
  By Assumption~\ref{thm:theory_a4}, every excursion is transient and is followed by recovery to the nominal-confidence region. Lemma~\ref{thm:theory_lemma2} bounds the state throughout each excursion, while Lemma~\ref{thm:theory_lemma1} guarantees strict Lyapunov decrease during the subsequent nominal phase whenever $\lVert x_t \rVert \ge r$. Since excursion growth is uniformly bounded and every excursion is followed by recovery, repeated switching cannot accumulate into unbounded state growth. Therefore, there exists $C > 0$ such that
  
  \begin{equation*}
    \lVert x_t \rVert \leq C \;\; \forall t.
  \end{equation*}

\end{proof}

\end{document}